\documentclass[twoside,leqno,twocolumn]{article}

\usepackage[letterpaper]{geometry}

\usepackage{ltexpprt}

\newif\ifapx
\apxtrue 

\newif\ifarxiv
\arxivtrue 

\newcommand{\ourmaintitle}{Estimating Mutual Information via Geodesic $k$NN}

\newcommand{\oururl}{\url{https://github.com/a-marx/geodesic-mi}}
\newcommand{\arxivid}{arXiv:2110.13883}

\usepackage{latexsym}
\usepackage[ruled, vlined, nofillcomment, linesnumbered, algo2e]{algorithm2e}
\usepackage[notheorem]{eda} 
\usepackage{bm} 
\usepackage{cancel} 
\usepackage{tikz}
\usepackage{amsfonts}
\usepackage{footmisc}
\usepackage{subcaption}

\usepackage[pdftitle={ourtitle}, hidelinks, hyperfootnotes = false]{hyperref}

\graphicspath{ {./figs/} }

\ifpdf
\usetikzlibrary{external,angles,arrows.meta,positioning}
\fi

\SetKwComment{tcpas}{\{}{\}}
\SetCommentSty{textnormal}
\SetArgSty{textnormal}
\SetKw{False}{false}
\SetKw{True}{true}
\SetKw{Null}{null}
\SetKwInOut{Output}{output}
\SetKwInOut{Input}{input}
\SetKw{AND}{and}
\SetKw{OR}{or}
\SetKw{Continue}{continue}

\newcommand{\ksg}{\textsc{KSG}\xspace}

\newcommand{\ourmethod}{\textsc{G-KSG}\xspace}

\newcommand{\gknn}{\textsc{gKNN}\xspace}
\newcommand{\lnn}{\textsc{LNN}\xspace}

\newcommand{\fung}{\ensuremath{g}\xspace}
\DeclarePairedDelimiterX{\infdivx}[2]{(}{)}{%
  #1\;\delimsize\|\;#2%
}

\DeclareMathOperator*{\argmax}{\arg\!\max}

\newcommand{\CSS}{\ensuremath{\mathit{css}}\xspace}
\newcommand{\CS}{\ensuremath{\mathit{cs}}\xspace}
\newcommand{\dgf}{\ensuremath{d_G}\xspace}
\newcommand{\Xt}{\ensuremath{\tilde{\bm{X}}}\xspace}
\newcommand{\Yt}{\ensuremath{\tilde{\bm{Y}}}\xspace}
\newcommand{\Zt}{\ensuremath{\tilde{\bm{Z}}}\xspace}
\newcommand{\Xg}{\ensuremath{\bm{X}}\xspace}
\newcommand{\Yg}{\ensuremath{\bm{Y}}\xspace}
\newcommand{\Zg}{\ensuremath{\bm{Z}}\xspace}
\newcommand{\Xn}{\ensuremath{\bm{E}_X}\xspace}
\newcommand{\Yn}{\ensuremath{\bm{E}_Y}\xspace}

\newcommand{\xg}{\ensuremath{\bm{x}}\xspace}
\newcommand{\yg}{\ensuremath{\bm{y}}\xspace}
\newcommand{\zg}{\ensuremath{\bm{z}}\xspace}

\newcommand{\Indep}{\mathop{\perp\!\!\!\perp}\nolimits}

\newtheorem{model}{Model}

\tikzset{node/.style={black, draw=black, circle, minimum size=1cm, scale=0.7}} 
\tikzset{snode/.style={black, draw=black, fill=lightgray, circle, minimum size=0.6cm, scale=0.7}} 
\tikzset{dummy/.style={black, draw=black, circle, minimum size=0.55cm, scale=0.7}} 
\tikzset{latent/.style={black, draw=black, fill=lightgray, circle, minimum size=1cm, scale=0.7}} 

\tikzset{causes/.style={->,very thick,  color=black}} 
\tikzset{causesxor/.style={->,very thick, dashed, color=black}} 
\tikzset{causesxoro/.style={o->,very thick, dashed, color=black}} 
\tikzset{connected/.style={o-o,very thick, color=black}} 
\tikzset{connectedd/.style={o-o,very thick, dashed,  color=black}} 
\tikzset{ocauses/.style={o->,very thick,  color=black}} 
\tikzset{confounder/.style={<->,very thick,  color=black}} 
\tikzset{confounderxor/.style={<->,very thick, dashed, color=black}} 
\tikzset{confounderl/.style={<->,very thick,  color=black, bend left=45}} 
\tikzset{confounderr/.style={<->,very thick,  color=black, bend right=45}} 

\tikzstyle{flatlabel}  = [above, font = \tiny, inner sep = 1pt, text = black]
\tikzstyle{flatlabelb}  = [below, font = \tiny, inner sep = 1pt, text = black]
\tikzstyle{slopelabel}  = [sloped, above, font = \tiny, inner sep = 1pt, text = black]
\tikzstyle{slopelabelb}  = [sloped, below, font = \tiny, inner sep = 1pt, text = black]

\begin{document}
\setlength{\pdfpagewidth}{8.5in}
\setlength{\pdfpageheight}{11in}

\title{\ourmaintitle}
\author{Alexander Marx \thanks{Department of Computer Science, ETH Zurich and ETH AI Center, Zurich, Switzerland. alexander.marx@ai.ethz.ch}
\and Jonas Fischer \thanks{Max Planck Institute for Informatics, Saarbr\"ucken, Germany. fischer@mpi-inf.mpg.de}}

\date{}

\maketitle

\fancyfoot[R]{\scriptsize{Copyright \textcopyright\ 2022 by SIAM\\
Unauthorized reproduction of this article is prohibited.}}
   
\begin{abstract}
\small\baselineskip=9pt 
Estimating mutual information (MI) between two continuous random variables $X$ and $Y$ allows to capture non-linear dependencies between them, non-parametrically. As such, MI estimation lies at the core of many data science applications. Yet, robustly estimating MI for high-dimensional $X$ and $Y$ is still an open research question.

In this paper, we formulate this problem through the lens of manifold learning. That is, we leverage the common assumption that the information of $X$ and $Y$ is captured by a low-dimensional manifold embedded in the observed high-dimensional space and transfer it to MI estimation. As an extension to state-of-the-art $k$NN estimators, we propose to determine the $k$-nearest neighbors via geodesic distances on this manifold rather than from the ambient space, which allows us to estimate MI even in the high-dimensional setting. An empirical evaluation of our method, \ourmethod, against the state-of-the-art shows that it yields good estimations of MI in classical benchmark and manifold tasks, even for high dimensional datasets, which none of the existing methods can provide.

\end{abstract}

\section{Introduction}
\label{sec:intro}

Quantifying the strength of a dependence between two continuous random variables is an essential task in data science~\cite{wang:09:divergence}. Due to its non-parametric nature, and hence its ability to measure complex non-linear dependencies, mutual information is ideal for this task~\cite{cover:06:elements}, which is why it is routinely applied for challenging settings such as gene network inference. 

Given two multidimensional continuous random variables $\Xg \in \mathbb{R}^{d_{\Xg}}$ and $\Yg \in \mathbb{R}^{d_{\Yg}}$, mutual information 
\begin{equation}
\label{eq:mi}
I(\Xg;\Yg) = h(\Xg) + h(\Yg) - h(\Xg,\Yg)
\end{equation}
can be expressed as sum of differential entropies
\begin{equation}
h(\Xg) = - \int_{\mathbb{R}^{d_{\Xg}}} f_{\Xg}(\xg) \log f_{\Xg}(\xg) d \xg \, ,
\end{equation}
where $\log$ refers to the natural logarithm. Although the differential entropy of a random variable can be negative, the chain rule does still apply for differential entropy, and hence $I(\Xg;\Yg) \ge 0$ with equality if and only if $\Xg$ is independent of $\Yg$~\cite{cover:06:elements}.

In an ideal scenario, where we are given an iid sample $(\xg_i, \yg_i)_{i=1,\dots,n} {\sim} f_{\Xg \Yg}$, and an unbiased estimator $\hat{f}_{\Xg \Yg}$, we could simply estimate $I(\Xg;\Yg)$ by individually estimating the differential entropies as $\hat{h}(\Xg) = - \frac{1}{n} \sum_{i=1}^n \log \hat{f}_{\Xg}(\xg_i)$.
In practice, such an unbiased density estimate is usually not available and state-of-the-art approaches instead resort to $k$NN based estimates of $\hat{h}(\Xg)$~\cite{frenzel:07:fp,kozachenko:87:firstnn}. Simply put, those methods estimate the log-density locally around each point, e.g.~by enclosing all its $k$NNs into a unit ball and computing its volume~\cite{grassberger:85:pre-ksg,kozachenko:87:firstnn}. 

This approach, i.e.~estimating each of the entropy terms individually with the same $k$, was, however, shown to induce a bias since the volume-related correction terms do not cancel~\cite{kraskov:04:ksg}.
To correct for this bias, Kraskov, St{\"{o}}gbauer, and Grassberger (\ksg)~\cite{kraskov:04:ksg} suggested to determine the distance to the $k$th neighbor only on the joint space and retrospectively count the data points falling within this region in $\Xg$ and $\Yg$. 
Other approaches try to reduce this bias by using different geometries that more tightly enclose the $k$NNs to better model the local densities~\cite{gao:15:strongly-dependent,gao:16:kernelandnn,lord:18:gknn,lu:20:gknn-boosting}.
Yet, none of these approaches has been successfully applied to high-dimensional data, which is exactly the setting we are interested in.

To estimate MI on high-dimensional data, we build upon the manifold assumption~\cite{cayton:05:manifold-learning-algorithms}. This common assumption in machine learning states that high-dimensional data resides on a low-dimensional manifold embedded in the ambient space. Under this assumption, we propose \ourmethod, which instantiates the \ksg estimator by determining the nearest neighbors via geodesic distances, i.e.~the shortest path between two points on the manifold they reside on. To estimate geodesic distances, we make use of and extend a recent proposal for manifold learning, called Geodesic Forests~\cite{madhyastha:20:gf}, which is an unsupervised random forest based on sparse linear projections.
As such, our method is well suited to estimate local densities and therewith mutual information on high-dimensional data implementing the manifold assumption.
Our main contributions are, we
\begin{itemize}
	\item establish a formal connection between manifold learning and mutual information estimation, for which we derive identifiability results in Sec.~\ref{sec:problem},
	\item propose \ourmethod, an instantiation of \ksg using geodesic distances, which we approximate via Geodesic Forests~\cite{cayton:05:manifold-learning-algorithms}, as explained in Sec~\ref{sec:geodesic},
	\item derive a locally adjusted dissimilarity measure, as well as a more efficient, $\mathcal{O}(n)$, split criterium for unsupervised forests in Sec~\ref{sec:geodesic}, and
	\item provide an extensive empirical evaluation of \ourmethod in Sec.~\ref{sec:exps}.
\end{itemize}

Before that, we discuss related work.%
\ifarxiv
\else
\!\footnote{Supplementary Material is published on \arxivid.}
\fi

\section{Related Work}\label{sec:related}

Mutual information estimation is a well studied problem for discrete, continuous and even discrete-continuous mixture data~\cite{marx:19:sci,valiant:11:sub-linear,gao:17:mixture,kozachenko:87:firstnn,kraskov:04:ksg,paninski:08:kernel,mandros:20:fdmixed,marx:21:myl}. Here, we focus on continuous data, for which a broad spectrum of MI estimators exists. Most common are estimators based on discretization~\cite{darbellay:99:adaptive-partitioning,marx:21:myl,koeman:14:cmi-forests}, kernel density estimation~\cite{paninski:08:kernel,gao:17:lnn}, and $k$-nearest neighbor estimation~\cite{frenzel:07:fp,kozachenko:87:firstnn,kraskov:04:ksg}.
Recently, $k$NN-based estimators have been established as state-of-the-art and can be computed efficiently, e.g.~via the $k$-D trie method~\cite{vejmelka:07:mi-fast-computation}.

Simply put, $k$NN-based methods estimate the local density around each point $i$ via its $k$-nearest neighbors~\cite{kraskov:04:ksg,lord:18:gknn}. Critical for the performance of these estimators are assumptions about the shape of the local volumes used to calculate the densities. One group of estimators measure the local distances via $L_2$ or $L_{\infty}$-norm~\cite{frenzel:07:fp,kozachenko:87:firstnn}. Other approaches try to estimate the volumes via locally computing an SVD~\cite{lord:18:gknn} or PCA~\cite{lu:20:gknn-boosting} transformation, or use a local Gaussian kernel~\cite{gao:17:lnn}. Alternatively, the \ksg~\cite{kraskov:04:ksg} estimator avoids estimating the volumes all along by computing the distance to the $k$th neighbor on the joint space, while simply counting the neighbors falling within this region in $\Xg$ and $\Yg$. Gao et al.~\cite{gao:18:demystifying} proved that the \ksg estimator is consistent and proposed a bias-corrected alternative, which focuses on low-dimensional data. Closest to our approach is $k$NN-based estimation using random forests to estimate the nearest neighbors, which requires either $\Xg$ or $\Yg$ to be discrete~\cite{mehta:19:cmi-forests-vogelstein}. None of these estimators has been evaluated on more than $20$ dimensions.

To efficiently estimate MI in a high-dimensional setting, we build upon the manifold assumption~\cite{cayton:05:manifold-learning-algorithms}, based on which embedding techniques were developed that successfully capture the most relevant information in few dimensions
by focussing on preserving local Euclidean distances and estimating geodesics that resemble the data location on the manifold~\cite{mcinnes:20:umap,vandermaaten:08:tsne,tenenbaum:00:isomap}.

In particular, we suggest a novel approach that estimates mutual information considering geodesic distances, combining ideas from manifold learning and MI estimation. We leverage recent advances of Madhyasta et al.~\cite{madhyastha:20:gf} in approximating geodesic distances based on tree estimates on sparse linear projections of the original space~\cite{breiman:01:rf,dasgupta:08:random-proj-trees}, which is suitable for $k$NN estimation on high-dimensional data.

\section{Preliminaries}
\label{sec:prelim}

Next, we briefly introduce the line of $k$NN based MI estimators more formally, and then shortly discuss the \ksg~\cite{kraskov:04:ksg} estimator and its limitations.

To estimate the differential entropy $\hat{h}(\Xg)$ of a random variable $\Xg$, $k$NN based estimators~\cite{frenzel:07:fp,kozachenko:87:firstnn} estimate the log-density locally around each point $\xg_i$, e.g.~by computing the volume $V_{\xg_i}$ of a ball enclosing all its $k$NNs~\cite{grassberger:85:pre-ksg,kozachenko:87:firstnn}. That is,
\begin{equation}
\label{eq:loghatf}
\log \hat{f}_{\Xg}(\xg_i) = \psi(k) - \psi(n) - \log V_{\xg} \, ,
\end{equation}
where $\psi(k)$ and $\psi(n)$ are the correction terms, with $\psi$ being the digamma function.\!\footnote{The digamma function is defined as $\psi(x) = \Gamma(x)^{-1}d\Gamma(x)/dx$. It satisfies the recursion $\psi(x+1)=\psi(x) + \frac{1}{x}$ with $\psi(1) = - C$, where $C=0.577215\dots$ is the Euler-Mascheroni constant.} A straight-forward approach to estimate $I(\Xg;\Yg)$ would be to estimate each of the involved entropy terms individually using Eq.~\eqref{eq:loghatf} based on the same $k$. 
Kraskov et al.~\cite{kraskov:04:ksg} showed that this will, however, induce a bias. Instead, they compute the distance to the $k$th neighbor only on the joint space and retrospectively count the data points falling within this region in $\Xg$ and $\Yg$.

\paragraph*{KSG Estimator}
Let $\Zg = (\Xg,\Yg)$ be the joint space spanned by $\Xg$ and $\Yg$. For any two data points $\zg_i$ and $\zg_j$, we define the distance between them as the maximum distance from their projections in $\Xg$ resp. $\Yg$,
\begin{equation}
\label{eq:max-norm}
d(\zg_i,\zg_j)_{\max} = \max \{ d(\xg_i, \xg_j), d(\yg_i, \yg_j) \} \, ,
\end{equation}
where the distances measured on the subspaces, i.e. $d(\xg_i, \xg_j)$ and $d(\yg_i, \yg_j)$, can be instantiated with any norm.
Furthermore, it is possible to use different norms for the subspaces $\Xg$ and $\Yg$~\cite{kraskov:04:ksg}.
In practice, both $d(\xg_i, \xg_j)$ and $d(\yg_i, \yg_j)$ are instantiated with the $L_{\infty}$-norm, and hence $d(\zg_i,\zg_j)_{\max}$ reduces to the $L_{\infty}$-norm over the joint space $(\Xg, \Yg)$.

Next, let $\frac{1}{2} \rho_{i,k}$ be the distance to the $k$-th neighbor on the $\Zg$ space using the maximum norm as defined above. We define the number of data points $\xg_j$ with a distance smaller than $\frac{1}{2} \rho_{i,k}$ to point the $\xg_i$ in the $\Xg$ subspace as $n_{\xg,i}$, i.e.
\begin{equation}
\label{eq:nx}
n_{\xg,i} = \left| \left\{ \xg_j \, : \, d(\xg_i, \xg_j) <  \frac{1}{2} \rho_{i,k}, \, i \neq j \right\} \right| \, ,
\end{equation}
and similarly, we define $n_{\yg,i}$ as the number of data points with a smaller distance to point $\yg_i$ than $\frac{1}{2} \rho_{i,k}$ on the $\Yg$ subspace. As an example consider the two-dimensional plot in Figure~\ref{fig:ksg-nx-ny}. In this example $n_{\xg,i} = 1$, that is, except $\xg_i$ itself, there exists only one further point with a distance $<  \frac{1}{2} \rho_{i,k}$ to $\xg_i$ for $k=1$. On the other hand, there exist $6$ data points, which fulfill this criterium for the $\Yg$ subspace. In general, it holds that $n_{\xg,i}+1 \ge k$, as well as $n_{\yg,i}+1 \ge k$.

\begin{figure}[t]
	\center
	\includegraphics{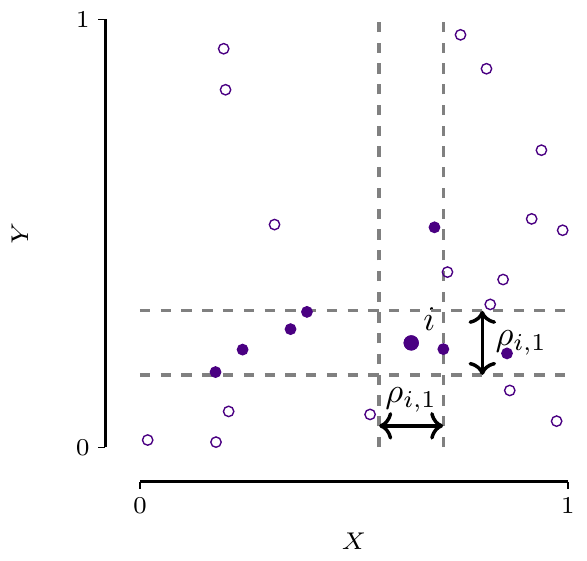}
	\caption{Example calculation for $\frac{1}{2}\rho_{i,k}$, where $k=1$, and we use the $L_{\infty}$-norm as distance measure for $\Xg$ and $\Yg$. In this case, $n_{\xg,i} = 1$ and $n_{\yg,i}=6$.}
	\label{fig:ksg-nx-ny}
\end{figure}

To compute the entropies for the subspaces $\Xg$ and $\Yg$, we consider the volumes related to $\rho_{i,k}$ computed on $\Zg$ and determine the corresponding $k$ retrospectively, i.e. we define $h^{\ksg}(\Xg)$ as
\begin{equation}
h^{\ksg}(\Xg) = - \frac{1}{n} \sum_{i=1}^n \psi(n_{\xg,i}+1) - \psi(n) - \log V_{\xg_i} \, ,
\end{equation}
where $\log V_{\xg_i}$ is computed as $\log c_{d_{\Xg}} - d_{\Xg} \rho_{i,k}$, with $c_{d_{\Xg}}$ being the volume of a $d_{\Xg}$-dimensional unit ball. By adding up the individual entropy terms, we arrive at the \ksg estimator, defined as
\begin{align}
\label{eq:ksg-estimator}
I^{\ksg}(\Xg;\Yg) &= \psi(k) {+} \psi(n) \\
&- \frac{1}{n} \sum_{i=1}^n \psi(n_{\xg,i}{+}1) {+} \psi(n_{\yg,i}{+}1) \, ,
\end{align}
where we can see that all volume terms cancel.

Although, it was shown that the \ksg estimator is consistent~\cite{gao:18:demystifying}, its bias increases with the number of dimensions $d$ as $\mathcal{O}(n^{-1/d})$~\cite{gao:18:demystifying}, which is problematic for the high-dimensional setting where $d > n$. In practice, we observe that the more dimensions we consider, the larger $n_{\xg,i}$ and $n_{\yg,i}$ become on average. This phenomenon occurs naturally, since we consider the maximum norm in the joint space. In extreme cases, $\psi(n_{\xg,i}{+}1) {+} \psi(n_{\yg,i}{+}1)$ is on average larger than $\psi(k) {+} \psi(n)$ and hence the estimate can be negative.

To address the limitations of \ksg in high-dimensional data, we leverage insights from manifold learning and mutual information estimation. In the next section, we will formally define the assumed data generative model and provide identifiability results.

\section{Mutual Information \& Manifold Learning}
\label{sec:problem}

The classical objective of MI estimation is to estimate $I(\Xg;\Yg)$ given an iid sample of the joint distribution $P_{\Xg,\Yg}$. Especially for high-dimensional and possibly noisy $\Xg$ and $\Yg$, estimating mutual information is challenging~\cite{lu:20:gknn-boosting,lord:18:gknn,gao:17:lnn}. Here, we view this problem from a manifold learning perspective~\cite{cayton:05:manifold-learning-algorithms},
where we assume that shared information between $\Xg$ and $\Yg$ is encoded in an intrinsic low-dimensional space $(\Xt,\Yt)$, whereas the majority of the dimensions of $\Xg$ and $\Yg$ are independent of each other, or correspond to noise dimensions, i.e. $I(\Xg;\Yg)$ is upper-bounded by $I(\Xt;\Yt)$.

To rigorously define the problem setting, we write down our assumptions about the data generative process as a structural causal model~\cite{pearl:09:causalitybook}.
Simply put, 
we assume that our observed variables $\Xg$ and $\Yg$ are both generated from a shared variable $\Zt$ and individual variables $\Xn$ and $\Yn$ that are independent of $\Zt$ and independent of each other (see Figure~\ref{fig:data-generation}). 
Further, we assume that the information that $\Xg$ contains about $\Zt$ is first passed through $\Xt$. Accordingly, the information that $\Yg$ has about $\Zt$ is processed through the path $\Zt \to \Yt \to \Yg$. We chose this model to reason about the shared information of $\Xg$ and $\Yg$ in terms of $I(\Xt;\Yt)$, which is assumed to be low-dimensional.

We formally define the generative model below.

\begin{model}
\label{model:1}
Given multi-dimensional random vectors $\Zt, \Xn, \Yn$, which are are pairwise independent. We generate $\Xg$ and $\Yg$ according to
\begin{align}
\Xt &= \fung_{\Xt}\left(\Zt \right)\qquad\quad\Yt = \fung_{\Yt}\left(\Zt \right) \\
\Xg &= \fung_{\Xg} \left(\Xt, \Xn \right)\quad\Yg = \fung_{\Yg} \left(\Yt, \Yn \right) \, ,
\end{align}
where we require that
\begin{enumerate}
	\item $\fung_{\Xt}$ and $\fung_{\Yt}$ preserve some information about $\Zt$, s.t.~$I(\Xt,\Yt) > 0$,
	\item $\fung_{\Xg}$, as well as $\fung_{\Yg}$ are required to be homeomorphisms (smooth and uniquely invertible maps), and
	\item for all sub-spaces $\Xt'$ or $\Yt'$ containing only a proper subset of rows of $\Xt$ resp. $\Yt$, it holds that $I(\Xt';\Yt) < I(\Xt; \Yt)$, resp. $I(\Xt; \Yt') < I(\Xt;\Yt)$.
\end{enumerate}
\end{model}

\begin{figure}[t]%
	\centering
	\includegraphics[]{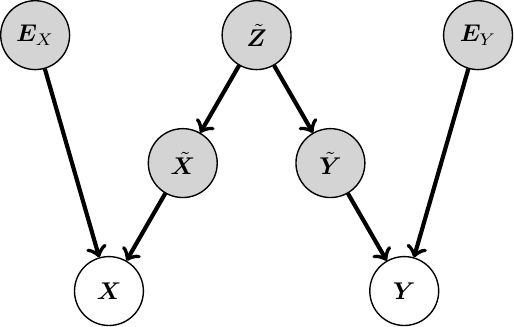}
	\caption{[Generative Model] We are interested in estimating $I(\Xt;\Yt)$, from $\Xg$ and $\Yg$. The random vector $\Xg$ is generated as a function of two unobserved  variables (denoted as shaded nodes) $\Xt$ and $\Xn$, where $\Xt \Indep \Xn$; $\Yg$ is generated accordingly. Further, $\Xt$ and $\Yt$ are generated from a shared latent factor $\Zt$.}
	\label{fig:data-generation}
\end{figure}

Conditions $1$-$3$ in Model~\ref{model:1} are very light requirements. If Cond.~1 would be violated, our estimator could still detect that there is no shared information between $\Xg$ and $\Yg$, but there would be little point in modeling the generative process from a manifold learning perspective. Similarly, Cond.~3 requires that all variables in the low-dimensional representation are contributing to the shared information about $\Xg$ and $\Yg$. Any subsets of variables not fulfilling this requirement would be modeled within the individual factor. Last, Cond.~2 assures that $I(\Xt;\Yt)$ can be recovered from $I(\Xg;\Yg)$, as shown in Proposition~\ref{prop:irrelevant-information}. This is a standard assumption in manifold learning: the low-dimensional manifold can be modeled as a homeomorphism of the ambient space or a subspace of it~\cite{cayton:05:manifold-learning-algorithms}.

In the following we will show that for a data generative process as defined in Model~\ref{model:1}, we can compute $I(\Xt;\Yt)$ given only $I(\Xg;\Yg)$. For arbitrary functions $\fung$ and $\fung'$, where the transformed variables $\fung(\Xg)$ and $\fung'(\Yg)$ are differentiable almost everywhere, the data processing inequality states that $I(\fung(\Xg);\fung'(\Yg)) \le I(\Xg;\Yg)$~\cite{cover:06:elements}. Applied to our scenario, $I(\Xg;\Yg)$ would at most provide a lower bound for $I(\Xt;\Yt)$. When requiring that $\fung$ and $\fung'$ are homeomorphisms, as stated in Condition~2 in Model~\ref{model:1}, we can make a stronger statement.

\begin{proposition}
\label{prop:irrelevant-information}
Given a data generative process as defined in Model~\ref{model:1}, then $I(\Xt;\Yt) = I(\Xg;\Yg)$.
\end{proposition}
\begin{proof}
Due to Condition~2 in Model~\ref{model:1}, we know that $\fung_{\Xg}$ and $\fung_{\Yg}$ are homeomorphisms. Kraskov et al.~\cite{kraskov:04:ksg} showed for any transformations $\Xg^t = \fung(\Xg)$, $\Yg^t = \fung'(\Yg)$, where $\fung$ and $\fung'$ are homeomorphisms, it holds that $I(\Xg;\Yg) = I(\Xg^t;\Yg^t)$. As an immediate consequence, we observe that in our case $I(\Xg;\Yg) = I(\Xt,\Xn;\Yt,\Yn)$. Hence, we can derive that
\begin{align}
I(\Xg;\Yg) &= h(\Xt,\Xn) - h \left( \Xt, \Xn \mid \Yt, \Yn \right) \\
&= h(\Xt, \Xn) - h\left(\Xt, \Xn \mid \Yt\right) \\
&= h(\Xt) + h(\Xn) - \left( h(\Xt \mid \Yt) + h(\Xn) \right) \\
&= h(\Xt) - h(\Xt \mid \Yt) = I(\Xt;\Yt) \, .
\end{align} 
In Line~2, we can omit $\Yn$ since $\Yn \Indep \Xn$, $\Yn \Indep \Zt$ by definition. Due to the Markov chain structure, $\Yn \Indep \Xt$, as well. Similarly, in Line~3 we can exploit that $\Xt \Indep \Xn$ and $\Yt \Indep \Xn$.
\end{proof}

Simply put, Proposition~\ref{prop:irrelevant-information} shows that, for a data generative process as defined in Model~\ref{model:1}, estimating $I(\Xg;\Yg)$ is equivalent to estimating $I(\Xt;\Yt)$. 
If we are provided only a finite number of iid samples, which is the case in practice, it allows us to significantly improve the sample efficiency for estimating $I(\Xg;\Yg)$. In particular, assume that we could recover the low-dimensional manifold $(\Xt, \Yt)$ with $\tilde{d}$ dimensions, which preserves the shared information on the high-dimensional space $(\Xg, \Yg)$ with $d \gg \tilde{d}$ dimensions. This could generally lead to a significant improvement of sample complexity for MI estimators. For the \ksg estimator, for example, the bias would decrease from
$\mathcal{O}\left(n^{-1/d} \right)$ to $\mathcal{O}\left(n^{-1 / \tilde{d}} \right)$, and hence only depend on the relevant dimensions.

In practice, this insight suggests a two-step procedure. In the first step, we aim to learn the low-dimensional manifold $(\Xt, \Yt)$ from $(\Xg, \Yg)$, and in the second step, we estimate the mutual information between $\Xg$ and $\Yg$ from the distances on the learned manifold $(\Xt, \Yt)$. In the next section, we propose such an approach based on Geodesic Forests.

\section{Geodesic Mutual Information Estimation}
\label{sec:geodesic}

In this section, we explain how to estimate mutual information via geodesic distances, where we first introduce our new estimator and then explain how we compute the corresponding quantities via Geodesic Forests.

\subsection{Geodesic \ksg}
\label{sec:gksg}

To efficiently estimate mutual information from data generated according to Model~\ref{model:1}, we propose to first learn the embedded low-dimensional manifold of the joint space $(\Xg,\Yg)$ via Geodesic Forests~\cite{madhyastha:20:gf}, as described subsequently in Sec~\ref{sec:gf}, and then compute the local $k$NN distances from this representation. In other words, we aim to approximate the distance between two points via the length of its shortest path on the manifold, i.e. its geodesic distance. 

More specifically, we approximate $d(\xg_i,\xg_j)$ and $d(\yg_i,\yg_j)$ in Eq.~\ref{eq:max-norm} with geodesic distances $\dgf(\xg_i,\xg_j)$ and $\dgf(\yg_i,\yg_j)$ obtained via marginalization (see Sec.~\ref{sec:geodesic-dist}) from the Geodesic Forest trained on $(\Xg,\Yg)$. To compute the distances on the joint space $\Zg = (\Xg,\Yg)$, we follow the \ksg approach~\cite{kraskov:04:ksg} and stick to the maximum between the distances on $\Xg$ and $\Yg$. Thus, we define $\frac{1}{2} \rho_{i,k}^G$ as the distance to the $k$th neighbor on the joint space $\Zg$ and obtain $n_{\xg,i}^G$ as
\begin{equation}
\label{eq:nxg}
n_{\xg,i}^G = \left| \left\{ \xg_j \, : \, \dgf(\xg_i, \xg_j) <  \frac{1}{2} \rho_{i,k}^G, \, i \neq j \right\} \right| \, ,
\end{equation}
and compute $n_{\yg,i}^G$ accordingly. Finally, we derive our proposed \ourmethod estimator as
\begin{align}
\label{eq:gksg-estimator}
I^{\ourmethod}(\Xg;\Yg) &= \psi(k) {+} \psi(n) \\
&- \frac{1}{n} \sum_{i=1}^n \psi(n_{\xg,i}^G{+}1) {+} \psi(n_{\yg,i}^G{+}1) \, .
\end{align}
Next, we explain how to estimate geodesic distances $\dgf$ from Geodesic Forests. 

\subsection{Geodesic Forests}
\label{sec:gf}

The term \emph{Geodesic Forest} (GF) has been introduced by Madhyastha et al.~\cite{madhyastha:20:gf} and describes an unsupervised version of sparse projection oblique randomer forests~\cite{tomita:20:sporf}. In a nutshell, each node of a tree in a GF is split based on a sparse linear projection of each data point onto a one-dimensional feature.
Classical splitting criteria allow to compute binary splits of the projected samples in this 1D space efficiently.
For a collection of trees, the relative geodesic similarity of two data points $\xg_i,\xg_j$ is estimated by the fraction of leaves they occur in \emph{together},
which has been proven successful to estimate geodesic distances even in the presence of many noise dimensions~\cite{madhyastha:20:gf}.

More formally, given a sample $\xg^n = \{ \xg_i, \dots, \xg_n \}$ of a $d_{\Xg}$-dimensional random vector $\Xg$, GF builds a set of $T$ trees. Each tree $T_i$ is trained on a bootstrapped sub-sample of size $m < n$, as typical for learning random forests~\cite{breiman:01:rf}. To grow a tree, we recursively split each parent node into its two child nodes until a certain stopping criterium is met. The two critical features, in which Geodesic Forests are different from classical random forests are the node splitting and the stopping criterium, which we describe in more detail below.

\paragraph*{Node Splitting}

Instead of splitting on a random feature, GF computes $p$ sparse random projections of the feature space and splits on that 1D projection, which minimizes the fast-BIC criterium (see below). To generate sparse projections, GF samples a random projection matrix $\bm{A} \in \{ -1,0,1 \}^{d_{\Xg} \times p}$, where an entry $a_{ij}$ is non-zero with probability $\lambda$, i.e.~$P(a_{ij} = 1) = P(a_{ij} = -1) = \frac{\lambda}{2}$, and zero otherwise. The sparsity parameter $\lambda$ is typically set to $\lambda = \frac{1}{d_{\Xg}}$. Given projection matrix $\bm{A}$, the projected feature matrix is $\Xg' = \bm{A}^T\Xg$, from which we can extract $p$ one-dimensional features, which are each evaluated by the splitting criterion.

\paragraph*{Fast-Bic}

To find a cut-point in a one-dimensional vector, the authors of GF~\cite{madhyastha:20:gf} introduce fast-BIC,
which is a regularized version of the classical two-means criterium~\cite{dasgupta:08:random-proj-trees}. Both criteria induce a hard cluster assignment to either the left or right cluster.

The general Bayesian Information Coefficient (BIC) for a model $M$ with parameter vector $\theta_M$ of length $| \theta_M |$ can be written as $\mathit{BIC}(M) = - 2 \log \hat{L} + \log (n) | \theta_M |,$
where $\log \hat{L}$ is the empirical log-likelihood of the data given model $M$. In our case, the model consists of five parameters, the cluster assignment and the parameters $\hat{\mu}_i$ and $\hat{\sigma}_i^2$, for $i \in \{ 1,2 \}$, which parameterize the assumed Gaussian distribution for cluster $i$.
Accordingly, the empirical negative log likelihood is defined as
\begin{equation}
- \log \hat{L} = \sum_{i=1}^2 \frac{n_i}{2} \left( 2 \log w_i - \log 2 \pi \hat{\sigma}_i^2 \right) \, ,
\end{equation}
where $w_i = n_i / n$ is the probability of a data point being assigned to cluster $i$.
Given an ordered one-dimensional sequence $x_1, \dots, x_n$, such that $x_i \le x_{i+1}$, we can compute the optimal split point in $\mathcal{O}(n^2)$ time, since we need to obtain for each of the potential $n-1$ cut-points the variances for both clusters.
We can, however, compute it even faster.

\paragraph*{Faster Fast-Bic} For an ordered sequence, we can reduce the runtime complexity, to determine the best split point, to $\mathcal{O}(n)$. Hence, even for an unordered sequence, we obtain a runtime in $\mathcal{O}(n \log n)$ by sorting, which is still faster than the original fast-BIC computation.
To achieve this speed-up, we use an elegant trick developed for segmentation. As derived by Terzi~\cite[Ch.~2]{terzi:06:segmentation}, we can compute the empirical variance $\hat{\sigma}^2_{i,j}$ for an arbitrary segment $x_i, \dots, x_j$, with $1 \le i \le j \le n$ as
\begin{equation}
\hat{\sigma}^2_{i,j} = \frac{1}{j{-}i{+}1} \left( \left( \CSS_j {-} \CSS_{i-1} \right) {-} \frac{1}{j{-}i{+}1} \left( \CS_j {-} \CS_{i-1} \right)^2 \right) \, ,
\end{equation}
where $\CS_i = \sum_1^i x_i$ is the cumulative sum of the first $i$ entries of the ordered sequence $x_1, \dots, x_n$ and $\CSS_i = \sum_1^i x_i^2$ the corresponding sum of squares, with $\CSS_0 =\CS_0 = 0$. In other words, after precomputing $\CS$ and $\CSS$ in linear time, we can compute the variance for an arbitrary segment in constant time.

\subsection{Approximate Geodesic Distances}
\label{sec:geodesic-dist}

To obtain a dissimilarity measure from a random forest, we can utilize the proximity score for random forests proposed by Breiman~\cite{breiman:01:rf}. That is, let $L_{ij}$ denote the number of trees for which data points $i$ and $j$ end up in the same leaf and let $T$ be the number of trees, the proximity score between two data points $i$ and $j$ is defined as
$p^F(\xg_i,\xg_j) = L_{ij}/T$.
Since $p^F(\xg_i,\xg_j) \in [0,1]$, we can compute a dissimilarity between two points as $d_F(\xg_i,\xg_j) = 1- p^F(\xg_i,\xg_j)$. In their empirical evaluation, Madhyastha et al.~\cite{madhyastha:20:gf} demonstrate that $d_F$ robustly recovers geodesic neighborhoods from high-dimensional data with many independent or noise dimensions.
As geodesic nearest neighbors, they define points that lie close on the low-dimensional manifold.

In theory, estimating the geodesic $k$-nearest neighbors is exactly what we are after, however, $d_F$ is not a proper distance metric. In particular, $d_F$ satisfies the reflexivity property ($d_F(\xg_i,\xg_i) = 0$), the non-negativity property, and the symmetry property ($d_F(\xg_i,\xg_j) = d_F(\xg_j,\xg_i)$). Despite those, $d_F(\xg_i,\xg_j) = 0$ does not imply that $\xg_i = \xg_j$, since two points could always end up in the same leaf, especially for small forests. Thus, the definiteness property is violated. In addition, $d_F$ does not satisfy the triangle inequality, i.e.~$d_F(\xg_i,\xg_k) \le d_F(\xg_i,\xg_j) + d_F(\xg_j,\xg_k)$ cannot be guaranteed. 

Hence, we propose an adjusted dissimilarity measure. The key idea of this adjusted measure builds upon the fact that locally a manifold resembles a Euclidean space, thus the geodesic distances locally become the $L_2$-norm. In our context, we assume that two points are close on the manifold, if $d_F(\xg_i,\xg_i)=0$. Under this premise, we propose the distance measure $\dgf$, i.e.
\begin{equation}
\label{eq:dgf}
\dgf(\xg_i,\xg_j) = 
\begin{cases}
d_F(\xg_i, \xg_j) & \text{if $d_F(\xg_i, \xg_j) > 0$,} \\
\frac{d_2(\xg_i, \xg_j)}{c  (T + \epsilon)} & \text{otherwise.}
\end{cases}
\end{equation}
In short, for all pairs $i, j$ for which $d_F(\xg_i, \xg_j) = 0$, we approximate their local geodesic distance via their Euclidian distance normalized by a constant factor. The normalization factor ensures that the normalized distances are always smaller than $\frac{1}{T}$, i.e.~the smallest non-zero value that $d_F$ can attain. It consists of $c$, the maximum $L_2$-norm between any two pairs $i,j$, the number of trees $T$ and a small constant $\epsilon > 0$.

As a result, $\dgf$ satisfies reflexivity, non-negativity and symmetry, and in addition, satisfies definiteness and locally (for those points, for which $d_F(\xg_i, \xg_j) = 0$) also satisfies the triangle inequality. We argue that possible violations of the triangle inequality for data points, for which $d_F(\xg_i, \xg_j) > 0$, are on average not relevant for $k$NN estimation with small $k \le 10$ in a high-dimensional setting, which is our main use-case.

\paragraph*{Marginal Distances}

Next, we briefly outline how we can compute the distance between two points $i,j$ on a subspace $\bm{S}$ of $\Xg$ given the forest learned on $\Xg$. 

In essence, we can compute $\dgf(\bm{s}_i,\bm{s}_j)$ in a straight forward manner. To compute the local distances, i.e. $\frac{d_2(\bm{s}_i,\bm{s}_j)}{c  (T + \epsilon)}$, we need to set $c$ to refer to the maximum distance between two points in the subspace $\bm{S}$, and recompute $d_F(\bm{s}_i,\bm{s}_j)$ for each pair $i,j$. To compute $d_F$ for a subspace $\bm{S}$, we first need to recompute all leave assignments. That is, given a tree $T_j$, we assign each point $i$ to that leaf in $T_j$, to which it would be assigned, 
if projected onto $\bm{S}$.
After reassigning the leaves, we can compute $d_F(\bm{s}_i,\bm{s}_j)$ as above.
Based on the above procedure, we can compute $\dgf(\xg_i,\xg_j)$ and $\dgf(\yg_i,\yg_j)$ for each pair $i,j$ and use these distances to compute our \ourmethod estimator as described in Sec.~\ref{sec:gksg}.

Next, we will empirically evaluate \ourmethod.

\section{Experiments}
\label{sec:exps}

We extensively evaluate \ourmethod against state-of-the-art MI estimators \ksg, \gknn, and \lnn~\cite{kraskov:04:ksg,lord:18:gknn,gao:17:lnn}.
\gknn uses ellipsoids computed from principal components to better fit the local data distribution, \lnn uses KDE with bandwidths automatically determined from the nearest neighbors. 
In particular, we compare on standard synthetic benchmark data as well as two simulated manifolds with known baseline MI for varying sample sizes, dimensionality of the data, and neighborhood size for each estimator.
For all experiments, we train geodesic forests with original parameters \cite{madhyastha:20:gf}, i.e.~$\lambda = 1/d$ number of dimensions, $T=300$ trees, and $\sqrt{2n}$ minimum number of points to split a node.
We use hyperparameters as suggested by the respective methods, details on which can be found in Supplementary Material~\ref{apx:params}, and
report the average across MI estimates of 20 repetitions for all experiments.
For reproducibility, we make code and data publicly available.\!\footnote{\oururl}

\subsection{Synthetic Data}

\begin{figure}[t]
	\begin{minipage}[t]{0.5\linewidth}
	\centering
	\includegraphics[]{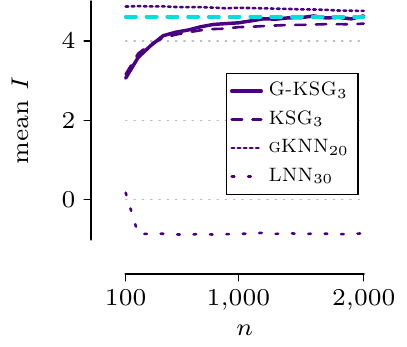}
	\end{minipage}%
	\begin{minipage}[t]{0.5\linewidth}
	\centering
	\includegraphics[]{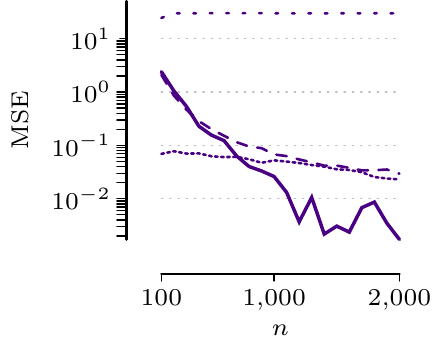}
	\end{minipage}%
	\caption{MI estimates with increasing sample size on uniform data (no noise), true MI is dashed cyan line.}
	\label{fig:accuracy-n100}
\end{figure}

We first evaluate the performance for estimating the MI between two variables generated from simple distributions, uniform and Gaussian.

\paragraph*{Sample Efficiency} To measure how well MI can be estimated with respect to the sample size, we draw datasets of size 100 to 2000 of $X$ and $Y$, where $X$ is uniformly distributed between $0$ and $1$ and $Y=X + N$ with $N \sim \mathit{Unif}(-\alpha/2,\alpha/2)$ and $\alpha = 0.01$.
The ground truth MI is given by $I(X;Y) = h(Y) - h(Z) = \frac{\alpha}{2} - \log \alpha$~\cite[Ex.~8.3]{cover:06:elements}.
Note that no independent variables $Z$ are added to the data, yet. 
We observe that even without independent variables in the data, \ourmethod is able to more efficiently estimate the MI, with an order of magnitude lower mean squared error (MSE) than classical \ksg (see Fig. \ref{fig:accuracy-n100}).
We further see that \lnn greatly underestimates the true mutual information on this uniform data.
While efficient even for as few as $100$ samples, \gknn constantly overestimates the true MI slightly, showing an order of magnitude larger mean squared error for $n\geq 1000$ samples compared to \ourmethod.
Note that, despite its complexity, \ourmethod is only a factor $10$ slower than classical \ksg, regardless of samples size (see Supplementary Material~\ref{apx:experiments}).
For the rest of the experiments, we will use $n=500$ samples, which is the largest sample size where the original KSG could keep up with \ourmethod.

\paragraph*{Uniform}
For the same data as above, we now add an increasing number of $\{0,2,...,20\}$ dimensions each sampled from a standard normal distribution to the original data.
We report the results in Fig.~\ref{fig:lin-unif} (top right), where we can observe that while \gknn, \ksg, and \ourmethod yield good estimates of the true MI without additional dimensions, \lnn drastically overestimates the true MI, exponentially increasing with the number of dimensions. The predicted MI of classical \ksg quickly falls to $0$ for as few as 4 dimensions and is hence useless for this data.
While also decreasing slightly, \ourmethod yields the most stable prediction of MI with respect to dimensionality, and with the lowest error on 15 or more dimensions.
\gknn underestimates even the original data without additional dimensions, and overestimates for higher-dimensional data, having numerical issues already for $10$ additional dimensions.

\paragraph*{High-Dimensional Data}
For the same uniform distribution, we generate data and add $\{50, 100, ..,600\}$ independent dimensions, adding half of the dimensions to $X$ and the other half to $Y$.
As expected, while \ksg immediately estimates a MI of $0$, and is therefore not useful at all, the MI estimate of \ourmethod decreases slightly but then remains stable with an increasing number of (independent) dimensions in the data.
Both \gknn as well as \lnn fail to yield any result even for $50$ dimensions due to numerical errors.
We provide results on this experiment in Supplementary Material~\ref{apx:experiments}.
Next, we investigate the behavior for variables from a different distribution.

\paragraph*{Gaussian}
We generate synthetic data of Gaussian distributed random variables $X$ and $Y$  with zero mean, unit variance and covariance of $0.9$.
Consequently, with correlation $\rho$ between $X$ and $Y$ being $0.9$, the true MI can be calculated as $I(X;Y) = - \frac{1}{2} \log(1-\rho^2)$.
For varying neighborhood sizes for density estimation, we report the results in Fig.~\ref{fig:lin-unif} (top left) for an increasing number of independent dimensions added to the original data. We observe that both \gknn as well as \lnn provide a decent estimate on the simple Gaussian data without additional dimensions, but then greatly overestimate the true MI for as few as $5$ independent or noise dimensions.
\gknn does not yield meaningful MI estimates beyond $10$ additional dimensions.
Classical \ksg provides a robust estimate on the simple Gaussian data, but quickly deteriorates with increasing dimensions, underestimating the true MI capturing almost zero mutual information.
While \ourmethod also experiences the same effect, it does so much more slowly and consistently yields the lowest error in terms of the true MI.

\subsection{Simulation Study}

\begin{figure}
	\begin{subfigure}[t]{0.49\linewidth}
		\includegraphics[]{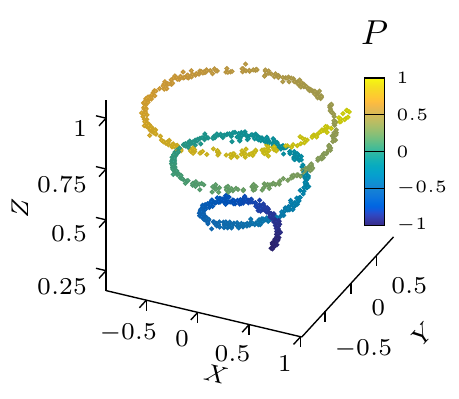}
	\end{subfigure}
	\hfill
	\begin{subfigure}[t]{0.49\linewidth}
		\includegraphics[]{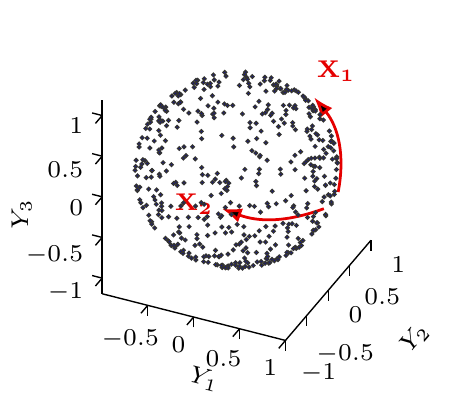}
	\end{subfigure}
	\caption{[Simulation Data] Shown are samples from the \texttt{Sphere} with input $(X_1,X_2)$ corresponding to (longitude, latitude) and output $(Y_1, Y_2, Y_3)$ cartesian coordinates (left), and \texttt{Helix} with input $P$ indicated by color, and cartesian output coordinates $(X,Y,Z)$ (right).}
	\label{fig:simdata}
\end{figure}

To test our approach for more complex data, and in particular data distributions resembling a manifold, we consider two simulated data sets.
The first dataset we study is a sphere, resembling e.g.~a planet (see Fig.~\ref{fig:simdata} left). Data is distributed on the sphere, for which $\Xg$ is given as longitude and latitude, and $\Yg$ is given as the 3D coordinates.
The second dataset is a helix (see Fig.~\ref{fig:simdata} right), for which $\Xg$ specifies the distance along the helix (1D), and $\Yg$ is given as the 3D coordinates, a typical problem from manifold learning.
We give details on how these datasets are sampled and how to compute a lower bound on the ground truth MI in Supplementary Material~\ref{apx:data-gen}.

\subsection{Helix}

First, we consider the \texttt{Helix} data.
Varying the number of independent dimensions for this data, we see that \ksg estimates that the data contains $0$ mutual information, as soon as any independent dimensions are added, whereas our method maintains a stable MI estimate even for many independent dimensions (see Fig. \ref{fig:lin-unif} bottom left). 
For \lnn we observe a sharp linear increase in predicted mutual information as a function of number of dimensions, first under and then over-predicting the true MI by a wide margin.
For \gknn, we observe first an overestimation of the MI which then comes close to the true MI for 2-6 additional dimensions. However, \gknn is not able to yield an estimate for more than $2$ additional dimensions when using $k=20$ respectively more than $6$ additional dimensions for $k=30$. Increasing the neighborhood size further does not make sense, as \gknn would then miss on the locality of the data when computing the volume, thus estimating global rather than local density.

\subsection{Sphere}

Next, we consider a simulation using the \texttt{Sphere} data, a simple dataset which, however, shows to be an astonishingly hard challenge for the state-of-the-art MI estimators.
We show the results in Fig.~\ref{fig:lin-unif} bottom right.
Overall, we observe similar trends as for \texttt{Helix}, giving further evidence that classical MI estimates are not able to capture manifolds within data with a large number of independent or noise dimensions.
As before, \ksg quickly deteriorates to estimate $0$ mutual information.
Similarly, \lnn shows a steep, near linear dependence between predicted MI and independent dimensions, which has little to do with the true MI.
Again, \gknn shows to perform poorly on the original task without added dimensions and fails to compute a result due to numerical issues on data of more than a handful of dimensions.
On this dataset, \ourmethod consistently predicts MI close to the ground truth, regardless dimensions.
Even for hundreds of added dimensions, \ourmethod remains stable estimates, whereas all other methods fail to do so (see Supplementary Material~\ref{apx:experiments}).

\begin{figure}[t]
	\begin{minipage}[t]{0.5\linewidth}
	\centering
	\includegraphics[]{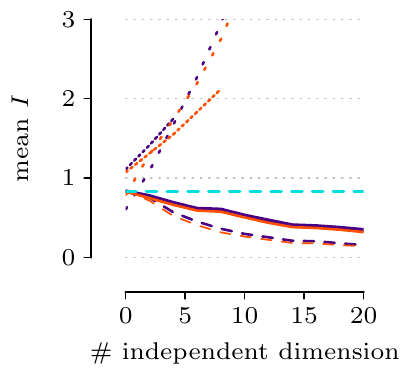}
	\end{minipage}%
	\begin{minipage}[t]{0.5\linewidth}
	\centering
	\includegraphics[]{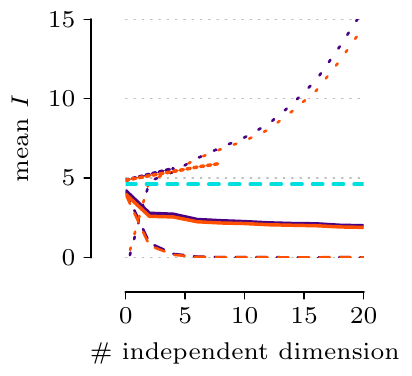}
	\end{minipage}%
	\linebreak
	\begin{minipage}[t]{0.5\linewidth}
	\centering
	\includegraphics[]{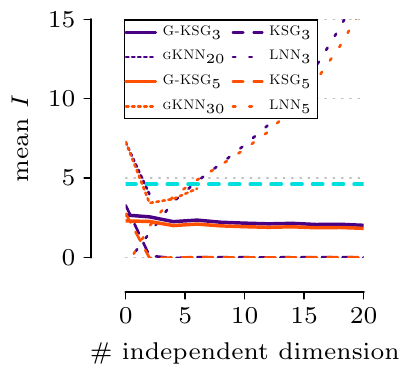}
	\end{minipage}%
	\begin{minipage}[t]{0.5\linewidth}
	\centering
	\includegraphics[]{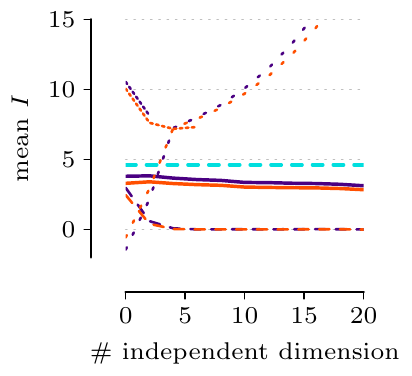}
	\end{minipage}%
	\caption{Mutual information estimates of \ourmethod, \ksg, \gknn, and \lnn for different number of $k$ indicated by method subscript for an increasing number of independent dimensions ($n{=}500$). Top-left Gaussian correlation, top-right uniform linear, bottom-left \texttt{Helix}, bottom-right \texttt{Sphere}, true MI is the dashed cyan line.}
	\label{fig:lin-unif}
\end{figure}

\section{Discussion \& Conclusion}
\label{sec:disc}

Mutual information, due to its non-parametric nature, captures non-linear dependencies and, hence, lends itself for measuring complex dependencies between random variables, and is routinely applied for challenging tasks such as gene network inference. 
Yet, estimation of MI on high-dimensional data remains an open problem.
In this work, we considered the problem of estimating MI for high-dimensional data under the manifold assumption, i.e. that the data has low intrinsic dimensionality.

To tackle this problem, we combined ideas from classical MI estimation and manifold learning.
In particular, we proposed to use geodesic distances to estimate local densities on the manifold rather than in the ambient space.
We leveraged the recently proposed Geodesic Forests, which can estimate relative geodesic distances even for high-dimensional data with many independent or noise dimensions.
To be able to scale to our setting, we further proposed a faster algorithm to compute the splitting criterion for individual nodes.
Based on computed geodesic distances, we then extended the state-of-the-art \ksg MI estimator to operate on these distances, leading to the \ourmethod estimator.

We evaluated \ourmethod against state-of-the-art MI estimators on standard benchmark data and two simulated manifold datasets.
The results show that \ourmethod outperforms its competitors, better estimating the true MI across tasks and settings.
Furthermore, it provides more stable results across data of different dimensionality, whereas the state-of-the-art fails to scale to higher dimensions, or greatly over- or underestimates the true MI. At the same time, \ourmethod is sample efficient, and is only slightly slower than classical \ksg, independent of sample size.

In summary, our approach allows us to efficiently and robustly estimate MI of high-dimensional data modeling the manifold assumption.
While thorough empirical evaluation showed that \ourmethod ably estimates the true MI across different datasets and dimensionalities,
it would make for engaging future work to study the theoretical guarantees of our estimator.
The estimation of geodesic distances renders this theoretical aspect extremely challenging.
Besides, we would be interested to study MI estimation for different settings, such as discrete-continuous mixtures or on time-series.

\section*{Acknowledgements}
AM is supported by the ETH AI Center. JF is supported by the International Max Planck Research School for Computer Science (IMPRS-CS).

\bibliographystyle{IEEEtranS}
\bibliography{bib/abbrev,bib/bib-jilles,bib/bib-paper}

\ifapx
\appendix
\section*{Supplementary Material}
\setcounter{section}{19}

The Supplementary Material, is split into three sections. In the first section, we provide the data generating mechanisms for the \texttt{Helix} and \texttt{Sphere} data, for which we derive lower bounds on the ground truth mutual information. Subsequently, we list the hyperparameter for \ourmethod and for the baselines in Section~\ref{apx:params}, and provide additional experiments in Section~\ref{apx:experiments}.

\subsection{Data Generation and MI Bounds}
\label{apx:data-gen}

In the following, we will first briefly recap the derivation of the ground truth of the mutual information for the linear uniform scenario and then explain the data generation, as well as, the derivation of the lower bounds for the ground truth for the \texttt{Helix} and \texttt{Sphere} data.

To derive the ground truth value for the linear uniform example, where $X {\sim} \mathit{Unif}(0,1)$, $Z {\sim}\mathit{Unif}(-\alpha/2,\alpha/2)$ and 
\[
Y = X + Z \, ,
\]
we follow Exercise~8.3 in~\cite{cover:06:elements}. First, note that
\begin{align}
I(X;Y) &= h(Y) - h(Y \mid X) \\
&= h(Y) - h(Z) = h(Y) - \log \alpha \, .
\end{align}
Thus, it remains to compute the entropy for $Y$, for which the density function can be computed as a convolution of $X$ and $Z$. Since $X$ and $Z$ follow a uniform distribution, $Y$ is a trapezoid and it can be derived that $h(Y) = \frac{\alpha}{2}$ for $\alpha \le 1$~\cite{cover:06:elements}.

\paragraph*{Helix}

For the \texttt{Helix} data, we generate a radius as $R\sim\mathit{Unif}(0,1)$, which we then embed as a spiral in a three-dimensional space. Accordingly, we generate the $X,Y,Z$ dimensions as
\begin{align*}
  P &= 5\pi + 3\pi R\,, \\
  X &= \frac{P \cos(P)}{8\pi} + N_1\,,\\
  Y &= \frac{P \sin(P)}{8\pi} + N_2\,,\\
  Z &= \frac{P}{8\pi} + N_3 \, ,
\end{align*}
where $N_1,N_2,N_3 \sim \mathit{Unif}(-\alpha/2,\alpha/2)$ are noise variables and $\alpha =0.01$.

To approximate the ground truth value of such a \texttt{Helix}, we can build upon the derivation for the uniform linear data to obtain a lower bound. A lower bound is sufficient in our case, since we do not measure how close we can estimate the true value, but how much information we can still recover in a noisy setting.

First, note that $I(X;X+Z) = I(X;X \sin (X) + Z)$ (see~\cite{gao:15:strongly-dependent}). More generally, $I(P;  \{ X,Y,Z \}) = I(P;  \{ X',Y',Z' \})$, where $X' = P + N_1$, $Y' = P + N_2$ and $Z' = P + N_3$. We get that
\begin{align}
I(P; \{ X',Y',Z' \}) &= h(X',Y',Z') - h(X',Y',Z' \mid P) \\
&= h(X',Y',Z') - h(N_1,N_2,N_3) \\
&= h(X') + h(Y' \mid X') + h(Z' \mid X',Y') \\
&- h(N_1,N_2,N3) \, ,
\end{align}
where again, we can rewrite $h(X',Y',Z' \mid P)$ as $h(N_1,N_2,N3)$ similar to the linear case. Additionally, all noise terms are independent of each other and thus $h(N_1,N_2,N3) = h(N_1) + h(N_2) + h(N_3) = 3 \log \alpha$. Now, due to the Markov chain structure, we can only approximate $h(Y' \mid X') \ge h(Y' \mid P)$. However, we conjecture that this approximation is quite close due to the low amount of noise added in the data generation. Similarly, $h(Z' \mid X',Y') \ge h(Z' \mid P)$. Thus
\begin{align}
 I(P; \{ X',Y',Z' \}) &\ge h(X') {+} h(Y' \mid P) {+} h(Z' \mid P) {-} 3 \log \alpha \\
&= h(X') - \log \alpha \\
&=  \frac{\alpha}{2} - \log \alpha \, .
\end{align}

\paragraph*{Sphere}

We generate samples from a \texttt{Sphere} by drawing latitude and longitude as $X_1,X_2\sim \mathit{Unif}(0,1)$ and compute the cartesian coordinates as
\begin{align*}
  Y_1 =& \cos(X_1) \cos(X_2) R + N_1 \,,\\
  Y_2 =& \cos(X_1) \sin(X_2) R + N_2 \,,\\
  Y_3 =& \sin(X_2) R + N_3 \,.
\end{align*}
where $N_1,N_2,N_3 \sim \mathit{Unif}(-\alpha/2,\alpha/2)$, and we use a radius of $R=1$ and $\alpha = 0.01$ for our experiments.

To derive a lower bound on the mutual information, we follow a similar procedure as for the \texttt{Helix} data. First, we can rewrite $Y_3' = X_2 + N_3$ and get that
\begin{align}
I( \bm{X}; \bm{Y}) &= h(Y_1,Y_2,Y_3') - h(Y_1,Y_2,Y_3' \mid X_1,X_2) \\
&= h(Y_3') + h(Y_2 \mid Y_3') + h(Y_1 \mid Y_2, Y_3') \\
&- h(Y_1,Y_2,Y_3' \mid X_1,X_2) \; ,
\end{align}
We can similarly decompose the conditional term to $h(Y_3' \mid X_1,X_2) + h(Y_2\mid X_1,X_2, Y_3' ) + h(Y_1\mid X_1,X_2, Y_3', Y_2)$. From the linear case, we know that $h(Y_3') - h(Y_3' \mid X_1,X_2) = \frac{\alpha}{2} - \log \alpha$, since $X_1 \Indep Y_3'$. Additionally, it is clear that $h(Y_2 \mid Y_3' ) \ge h(Y_2\mid X_1,X_2, Y_3' )$ and $h(Y_1\mid Y_3', Y_2) \ge h(Y_1\mid X_1,X_2, Y_3', Y_2)$ and thus $I(\bm{X} ;\bm{Y}) \ge \frac{\alpha}{2} - \log \alpha$.

\subsection{Hyperparameter}
\label{apx:params}
For \ourmethod, we use projections with $\lambda = 1/d$ number of dimensions, $T=300$ trees, and $\sqrt{2n}$ minimum number of points in a node to consider it for splitting. These parameters are suggested by the authors of geodesic forests~\cite{madhyastha:20:gf}.

In case of \lnn, we use the suggested setting, which is $k=30$ neighbours for MI estimation and evaluate $k'=3$, as well as $k'=5$, for bandwidth prediction for the kernel density estimation.

For \gknn and \ksg we set the neighbourhood size as specified in the main paper.

\subsection{Additional Experiments}
\label{apx:experiments}

Next, we supplement the runtimes for the experiment on linear uniform data with increasing sample size and provide the results for the experiments on high-dimensional data.

\paragraph*{Time comparison}

We provide a time comparison of all methods in Fig.~\ref{fig:runtime-n100} for 10 repetitions. We observe that \ourmethod only takes $10$ times longer thank \ksg, independent of number of samples.

\begin{figure}[t]
	\begin{minipage}[t]{\linewidth}
	\centering
	\includegraphics[]{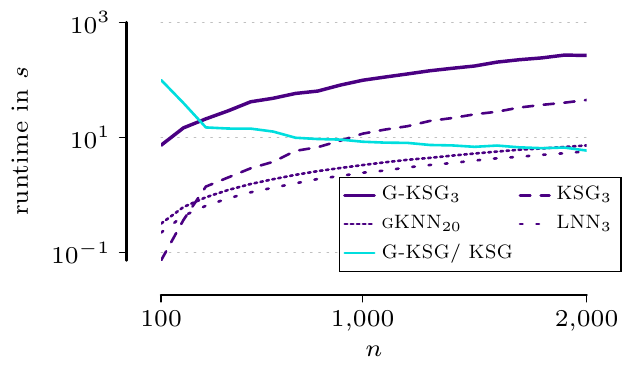}
	\end{minipage}%
	\caption{Shown is the runtime in seconds for \ourmethod, \ksg, \gknn and \lnn for the experiment on sample efficiency in Sec~6.1. In addition, we show the quotient of the runtime of \ourmethod divided by the runtime of \ksg. For larger sample sizes ($n \ge 1000$), \ksg is faster by a factor smaller than $10$.}
	\label{fig:runtime-n100}
\end{figure}

\paragraph*{High dimensional results}
We provide the results for high dimensional data for the uniform linear data and \texttt{Sphere} across 3 repetitions in Fig.~\ref{fig:results-high-dimensional}.
We varied the number of independent dimensions in $\{0, 100,200,...,600\}$, splitting equally between $X$ and $Y$. Both \gknn as well as \lnn were not able to compute the MI estimate due to numerical issues even for as few as 50 dimensions.

\begin{figure}[t]
	\begin{minipage}[t]{\linewidth}
	\centering
	\includegraphics[]{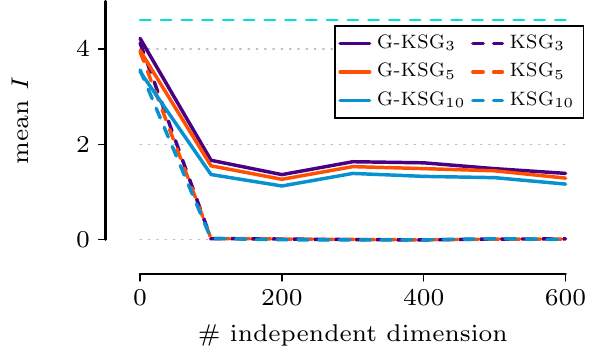}
	\end{minipage}%
	\linebreak
	\begin{minipage}[t]{\linewidth}
	\centering
	\includegraphics[]{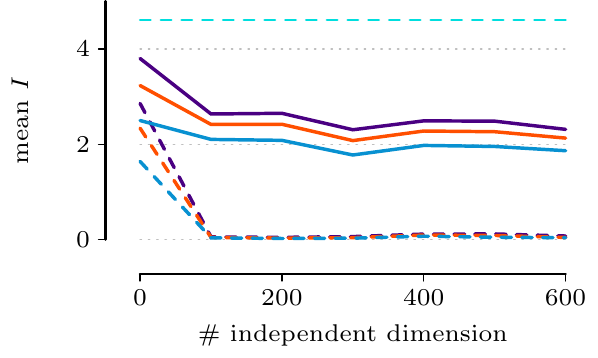}
	\end{minipage}%
	\caption{Mutual information estimates of \ourmethod and \ksg $k=\{3,5,10\}$ on linear data with uniform source (top) and sphere data (bottom) with an increasing number of noise dimensions up to the high-dimensional setting ($d>n=500$). The dashed cyan line is a lower bound on the ground truth MI.}
	\label{fig:results-high-dimensional}
\end{figure}
\fi

\end{document}